\documentclass[11pt, twoside, letterpaper]{amsart}
\usepackage{amsmath, hyperref, color}
\usepackage{srcltx}

\usepackage{geometry}
\newtheorem{thm}{Theorem}[section]

\newtheorem{lem}[thm]{Lemma}
\newtheorem{prop}[thm]{Proposition}

\theoremstyle{definition}

\newtheorem{ass}[thm]{Assumption}

\theoremstyle{remark}
\newtheorem{rem}[thm]{Remark}

\numberwithin{equation}{section}
\newcommand{\A}{\mathcal{A}}
\newcommand{\Z}{\mathcal{Z}}
\newcommand{\X}{\mathcal{X}}
\newcommand{\Y}{\mathcal{Y}}
\newcommand{\T}{\mathcal{T}}
\newcommand{\R}{\mathbb{R}}
\newcommand{\N}{\mathbb{N}}
\newcommand{\LL}{\mathbf{L}}
\newcommand{\expec}{\mathbb{E}}
\newcommand{\M}{\mathcal{M}}

\newcommand{\PP}{\mathbb{P}}
\newcommand{\QQ}{\mathbb{Q}}
\newcommand{\F}{\mathcal{F}}

\newcommand{\ud}{\mathrm d}

\newcommand{\dbracc}[1]{[\kern-0.15em[ #1 ]\kern-0.15em]}
\newcommand{\dbraco}[1]{[\kern-0.15em[ #1 [\kern-0.15em[}
\newcommand{\dbraoc}[1]{]\kern-0.15em] #1 ]\kern-0.15em]}
\newcommand{\dbraoo}[1]{]\kern-0.15em] #1 [\kern-0.15em[}
\newcommand{\be}{\begin{equation}}
\newcommand{\ee}{\end{equation}}
\newcommand{\ba}{\begin{aligned}}
\newcommand{\ea}{\end{aligned}}
\newcommand{\ind}{\mathbb{I}}

\DeclareMathOperator{\cl}{cl}
\DeclareMathOperator{\conv}{conv}

\begin{document}

\title[Optimal investment under NUPBR]{Optimal investment with intermediate consumption under no unbounded profit with bounded risk}%

\author{Huy N. Chau}%
\address{Huy N. Chau, Alfred Renyi Institute of Mathematics, Hungarian Academy of Sciences (Hungary)}%
\email{chau@renyi.hu}%

\author{Andrea Cosso}%
\address{Andrea Cosso, Dipartimento di Matematica, Politecnico di Milano (Italy)}%
\email{andrea.cosso@polimi.it}%

\author{Claudio Fontana}%
\address{Claudio Fontana, Laboratoire de Probabilit\'es et Mod\`eles Al\'eatoires, Paris Diderot University (France)}%
\email{fontana@math.univ-paris-diderot.fr}%

\author{Oleksii Mostovyi}

\address{Oleskii Mostovyi, Department of Mathematics, University of Connecticut (USA)}%
\email{oleksii.mostovyi@uconn.edu}%

\subjclass[2010]{91G10, 93E20. \textit{JEL Classification:} C60, G1.}
\keywords{Utility maximization, arbitrage of the first kind, local martingale deflator, duality theory, semimartingale, incomplete market}%

\date{\today}%


\maketitle

\begin{abstract}
We consider the problem of optimal investment with intermediate consumption in a general semimartingale model of an incomplete market, with preferences being represented by a utility stochastic field. 
We show that the key conclusions  of  the  utility  maximization  theory  hold  under  the   assumptions  of  no unbounded profit with bounded risk (NUPBR) and of the finiteness of both primal and dual value functions.
\end{abstract}

\section{Introduction}		\label{sec:intro}

Since the pioneering work of \cite{HarrisonKreps79}, equivalent (local/sigma) martingale measures play a prominent role in the problems of pricing and portfolio optimization. 
Their existence is equivalent to the absence of arbitrage in the sense of \emph{no free lunch with vanishing risk} (NFLVR) (see \cite{DelbaenSchachermayer94,DelbaenSchachermayer98b}), and this represents the standard no-arbitrage-type assumption in the classical duality approach to optimal investment problems (see e.g. \cite{KS99,KS2003,KaratzasZitkovic03,Zitkovic05}).
In a general semimartingale setting, necessary and sufficient conditions for the validity of the key assertions of the utility maximization theory (with the possibility of intermediate consumption) have been recently established in \cite{mostovyi2015}. More specifically,  such assertions have been proven in \cite{mostovyi2015} under the assumptions that the primal and dual value functions are finite and that there exists an \emph{equivalent martingale deflator}. In particular, in a finite time horizon, the latter assumption is equivalent to the validity of NFLVR.

In this paper, we consider a general semimartingale setting with an infinite time horizon where preferences are modeled via a utility stochastic field, allowing for intermediate consumption. Building on the abstract theorems of \cite{mostovyi2015}, our main result shows that the standard assertions of the utility maximization theory hold true as long as there is \emph{no unbounded profit with bounded risk} (NUPBR) and the primal and dual value functions are finite. In general, NUPBR is weaker than NFLVR and can be shown to be equivalent to the existence of an \emph{equivalent local martingale deflator}. 
Our results give a precise and general form to a widespread meta-theorem in the mathematical finance community stating that the key conclusions of the utility maximization theory hold under NUPBR. 
Even though such a result has been proven in some specific formulations of the utility maximization problem (see the discussion below), to the best of our knowledge, it has not been justified in general semimartingale settings with an arbitrary consumption clock and a stochastic Inada utility.

The proofs rely on certain characterizations of the dual feasible set. Thus, in Lemma~\ref{lem:cons} we give a polarity {description}, show its closedness under countable convex combinations in Lemma~\ref{lem:closedness}, and demonstrate in 
Proposition~\ref{prop:NUPBR}
that nonemptyness of the set that generates the dual domain is equivalent to NUPBR. Upon that, we prove the bipolar relations between primal and dual feasible sets and apply the abstract theorems from~\cite{mostovyi2015}. As an implication of the bipolar relations, we also show how Theorem 2.2 in \cite{KS99} 
can be extended to hold under NUBPR (instead of NFLVR), see Remark~\ref{bipolar} below for details.

Neither NFLVR, nor NUPBR by itself guarantee the existence of solutions to utility maximization problems, see \cite[Example 5.2]{KS99} and \cite[Example 4.3]{KasperChristensen} for counterexamples. This is why finiteness of the value functions is needed in the formulation of our main result. 
However, it is shown in \cite{CDM15} that NUPBR holds if and only if, for every sufficiently nice deterministic utility function, the problem of maximizing expected utility from terminal wealth admits a solution under an equivalent probability measure, which can be chosen to be arbitrarily close to the original measure (see \cite[Theorem 2.8]{CDM15} for details).
Besides, NUPBR represents the minimal no-arbitrage-type assumption that allows for the standard conclusions of the theory to hold for the problem of maximization of expected utility from terminal wealth. Indeed, by \cite[Proposition 4.19]{KaratzasKardaras07}, the failure of NUPBR implies that there exists a time horizon such that the corresponding utility maximization problem  either does not have a solution, or has infinitely~many.
Our work complements these papers by providing the convex duality results under NUPBR, also allowing for stochastic preferences as well as intermediate consumption.

The problem of utility maximization without relying on the existence of martingale measures has already been addressed in the literature. 
In the very first paper \cite{Merton69} on expected utility maximization in continuous time settings, an optimal investment problem is explicitly solved even though an equivalent martingale measure does not exist in general in the infinite time horizon case.
In an incomplete It\^o process setting under a finite time horizon, \cite{KLSX} have considered 
 the problem of maximization of expected utility from terminal wealth and established 
the existence results for an optimal portfolio via convex duality theory
   without the full strength of NFLVR (see also \cite[Section 10.3]{FK09} and \cite[Section 4.6.3]{FontanaRunggaldier13}). In particular, in view of \cite[Theorem 4]{Kardaras10}, Assumption 2.3 in \cite{KLSX} is {equivalent} to the nonemptyness of the set of equivalent local martingale deflators. 
Passing from an It\^o process to a continuous semimartingale setting, the results of \cite{KS99} 
have been extended by weakening the NFLVR requirement in \cite{Larsen09} (note that \cite[Assumption 2.1]{Larsen09} is equivalent to  NUPBR). 
In a general semimartingale setting, \cite{LarsenZitkovic13} have established convex duality results for the problem of maximizing expected utility from terminal wealth (for a deterministic utility function) in the presence of trading constraints without relying on the existence of martingale measures. In particular, in the absence of trading constraints, the no-arbitrage-type requirement adopted in \cite{LarsenZitkovic13} turns out to be equivalent to NUPBR. Indeed, {\cite[Assumption 2.3]{LarsenZitkovic13} requires the $\LL^0_+$-solid hull\footnote{As usual, $\LL^0$ denotes the space of equivalence classes of real-valued random variables on the probability space $(\Omega,\F,\PP)$, equipped with the topology of convergence in probability; $\LL^0_+$ is the positive orthant of $\LL^0$.} of the set of all terminal wealths generated by admissible strategies with initial wealth $x$,  denoted by $\mathcal{C}(x)$, to be convexly compact\footnote{By \cite[Theorem 3.1]{Zit09}, a closed convex subset of $\LL^0_+$ is convexly compact if and only if it is bounded in $\LL^0$.}
for all $x\in\R$ and nonempty for some $x\in\R$. In the absence of trading constraints, \cite[Theorem 2]{Kardaras10} shows that the boundedness in $\LL^0$ of $\mathcal{C}(x)$ already implies its closedness in $\LL^0$, thus in such a framework the convex compactness of $\mathcal{C}(x)$ 
holds if and only if the NUPBR condition does.}


%


The paper is structured as follows. 
Section \ref{sec:2} begins with a description of the general setting (Subsection \ref{sec:setting}), introduces and characterizes the NUPBR condition (Subsection \ref{sec:NUPBR}) and then proceeds with the statement of the main results (Subsection \ref{sec:invest}).  Section \ref{sec:proofs} contains the proofs of our results.

\section{Setting and main results}	\label{sec:2}

\subsection{Setting}	\label{sec:setting}

Let $(\Omega,\F,(\F_t)_{t\in[0,\infty)},\PP)$ be a complete stochastic basis, with $\F_0$ being the completion of the trivial $\sigma$-algebra, and $S=(S_t)_{t\geq0}$ an $\R^d$-valued semimartingale, representing the discounted prices of $d$ risky assets\footnote{As explained
in \cite[Remark 2.2]{mostovyi2015}, there is no loss of generality in assuming that asset prices are discounted, since we allow for preferences represented by utility stochastic fields (see Section~\ref{sec:invest} below).}. 
We fix a \emph{stochastic clock} $\kappa=(\kappa_t)_{t\geq0}$, which is a nondecreasing, c\`adl\`ag adapted process such that
\be	\label{clock}
\kappa_0 = 0,
\qquad
\PP(\kappa_{\infty}>0)>0
\qquad\text{and}\qquad
\kappa_{\infty}\leq A,
\ee
for some finite constant $A$. The stochastic clock $\kappa$ represents the notion of time according to which consumption is assumed to occur. By suitably specifying the clock process $\kappa$, several different formulations of investment problems, with or without intermediate consumption, can be recovered from the present setting (see \cite[Section 2.8]{Zitkovic05} and \cite[Examples 2.5-2.9]{mostovyi2015}). 

A \emph{portfolio} is defined by a triplet $\Pi=(x,H,c)$, where $x\in\R$ represents an initial capital, $H=(H_t)_{t\geq0}$ is an $\R^d$-valued predictable $S$-integrable process representing the holdings in the $d$ risky assets and $c=(c_t)_{t\geq0}$ is a nonnegative optional process representing the consumption rate. 
The discounted value process $V=(V_t)_{t\geq0}$ of a portfolio $\Pi=(x,H,c)$ is defined as
\[
V_t := x + \int_0^tH_u\,\ud S_u - \int_0^tc_u\,\ud \kappa_u,
\qquad t\geq0.
\]
We let $\X$ be the collection of all nonnegative value processes associated to portfolios of the form $\Pi=(1,H,0)$, i.e.,
\[
\X := \left\{X\geq0 : X_t=1+\int_0^tH_u\,\ud S_u,t\geq0\right\}.
\]
For a given initial capital $x>0$, a consumption process $c$ is said to be \emph{$x$-admissible} if there exists an $\R^d$-valued predictable $S$-integrable process $H$ such that the value process $V$ associated to the portfolio $\Pi=(x,H,c)$ is nonnegative. The set of $x$-admissible consumption processes corresponding to a stochastic clock $\kappa$ is denoted by $\A(x)$.
For brevity, we let $\A:=\A(1)$.

\subsection{No unbounded profit with bounded risk}	\label{sec:NUPBR}

In this paper, we shall assume the validity of the following no-arbitrage-type condition:
\be	\tag{NUPBR}	\label{NUPBR}
\text{the set } 
\X_T := \bigl\{X_T : X\in\mathcal{X}\bigr\} 
\text{ is bounded in probability, for every }T\in\R_+.
\ee
For each $T\in\R_+$, the boundedness in probability of the set $\X_T$ has been named \emph{no unbounded profit with bounded risk} in \cite{KaratzasKardaras07} and, as shown in \cite[Proposition 1]{Kardaras10}, is equivalent to the absence of \emph{arbitrages of the first kind} on $[0,T]$. Hence, condition \eqref{NUPBR} above is equivalent to the absence of arbitrages of the first kind in the sense of \cite[Definition 1]{Kar14}. 

We define the set of \emph{equivalent local martingale deflators (ELMD)} as follows:
\begin{align*}
\mathcal{Z} := \bigl\{ Z >0 :\;& Z \text{ is a c\`adl\`ag local martingale such that  } Z_0 =1 \text{ and } \\
& ZX = (Z_tX_t)_{t\geq 0} \text{ is a local martingale for every } X \in \mathcal{X} \bigr\}.
\end{align*}

The following result is already known in the one-dimensional case in a finite time horizon (see \cite[Theorem 2.1]{Kardaras12}). The extension to the multi-dimensional infinite horizon case relies on~\cite[Theorem 2.6]{TS14} (see also \cite[Proposition 2.3]{ACDJ14}).

\begin{prop}	\label{prop:NUPBR}
Condition \eqref{NUPBR} holds if and only if $\Z\neq\emptyset$.
\end{prop}


\begin{rem}	\label{rem:comp_mostovyi}
In \cite{mostovyi2015}, it is assumed that 
\begin{equation}\label{NFLVF}
\{Z\in\Z : Z\text{ is a martingale}\}\neq\emptyset,
\end{equation}
which is stronger than \eqref{NUPBR} by Proposition \ref{prop:NUPBR}. A classical example where \eqref{NUPBR} holds but \eqref{NFLVF} fails is provided by the three-dimensional Bessel process (see e.g. \cite{DelbaenSchachermayer95a}, \cite[Example 2.2]{Larsen09}, and \cite[Example 4.6]{KaratzasKardaras07}).
\end{rem}

\subsection{Optimal investment with intermediate consumption}	\label{sec:invest}
We now proceed to show that the key conclusions of the utility maximization theory can be established under condition \eqref{NUPBR}. 
We assume that preferences are represented by a \emph{utility stochastic field} $U=U(t,\omega,x):[0,\infty)\times\Omega\times[0,\infty)\rightarrow\R\cup\{-\infty\}$ satisfying the following assumption. 

\begin{ass}	\label{ass:U}
For every $(t,\omega)\in[0,\infty)\times\Omega$, the function $x\mapsto U(t,\omega,x)$ is strictly concave, strictly increasing, continuously differentiable on $(0,\infty)$ and satisfies the Inada conditions
\[
\underset{x\downarrow0}{\lim} \, U'(t,\omega,x) = +\infty
\qquad\text{and}\qquad
\underset{x\rightarrow+\infty}{\lim} \, U'(t,\omega,x) = 0,
\]
with $U'$ denoting the partial derivative of $U$ with respect to its third argument. By continuity, at $x=0$ we {suppose} that $U(t,\omega,0)=\lim_{x\downarrow0}U(t,\omega,x)$ (note that this value may be $-\infty$). Finally, for every $x\geq0$, the stochastic process $U(\cdot,\cdot,x)$ is optional.
\end{ass}

To a utility stochastic field $U$ satisfying Assumption \ref{ass:U}, we  associate the \emph{primal value function}, defined as
\be	\label{value}
u(x) := \sup_{c\in\A(x)}\expec\left[\int_0^{\infty}U(t,\omega,c_t)\,\ud\kappa_t\right],\quad x>0,
\ee
with the convention $\expec[\int_0^{\infty}U(t,\omega,c_t)\,\ud\kappa_t]:=-\infty$ if $\expec[\int_0^{\infty}U^-(t,\omega,c_t)\,\ud\kappa_t]=+\infty$.


In order to construct the dual value function, we define as follows the \emph{stochastic field $V$ conjugate to $U$}:
\[
V(t,\omega,y) := \sup_{x>0}\bigl(U(t,\omega,x)-xy\bigr),
\qquad (t,\omega,y)\in[0,\infty)\times\Omega\times[0,\infty).
\]
We also introduce the following set of dual processes:
\begin{align*}
\Y(y) := \cl\bigl\{Y :\; Y\text{ is c\`adl\`ag adapted and }
0\leq Y\leq yZ \text{ $(\ud\kappa\times\PP)$-a.e. for some }Z\in\Z\bigr\},
\end{align*}
where the closure is taken in the topology of convergence in measure $(\ud\kappa\times\PP)$ on the space of real-valued optional processes.  We write $\Y:=\Y(1)$ for brevity.
The value function of the dual optimization problem (\emph{dual value function}) is then defined as
\be	\label{dual_value}
v(y) := \inf_{Y\in\Y(y)}\expec\left[\int_0^{\infty}V(t,\omega,Y_t)\,\ud\kappa_t\right]{,\quad y>0},
\ee
with the convention $\expec[\int_0^{\infty}V(t,\omega,Y_t)\,\ud\kappa_t]:=+\infty$ if $\expec[\int_0^{\infty}V^+(t,\omega,Y_t)\,\ud\kappa_t]=+\infty$.
We are now in a position to state the following theorem, 
which is the main result of this paper.
\begin{thm}	\label{thm:main}
Assume that conditions \eqref{clock} and \eqref{NUPBR} hold true and let $U$ be a utility stochastic field satisfying Assumption \ref{ass:U}. Let us also suppose that
\be	\label{eq:finiteness}
v(y)<\infty\quad\text{ for every }y>0
\quad\text{and}\quad
u(x)>-\infty\quad\text{ for every }x>0.
\ee
Then the primal value function $u$ and the dual value function $v$ defined in \eqref{value} and \eqref{dual_value}, respectively, satisfy the following properties:
\begin{enumerate}
\item[(i)] 
$u(x)<\infty$, for every $x>0$, and $v(y)>-\infty$, for every $y>0$. The functions $u$ and $v$ are conjugate, i.e.,
\begin{displaymath}
v(y) = \underset{x>0}{\sup} \bigl(u(x)-xy\bigr),\quad y>0, \qquad
u(x) = \underset{y>0}{\inf} \bigl(v(y)+xy\bigr),\quad x>0;
\end{displaymath}
\item[(ii)]
the functions $u$ and $-v$ are continuously differentiable on $(0,\infty)$, strictly concave, strictly increasing and satisfy the Inada conditions
\begin{displaymath}
\begin{array}{rclccrcl}
\underset{x\downarrow0}{\lim} \, u'(x) & = &+\infty, &\quad&
\underset{y\downarrow0}{\lim} \, -v'(y) &=& +\infty,	\\
\underset{x\rightarrow+\infty}{\lim} \, u'(x) &=& 0, &\quad&
\underset{y\rightarrow+\infty}{\lim} \, -v'(y) &=& 0. \\
\end{array}
\end{displaymath}
\end{enumerate}
Moreover, for every $x>0$ and $y>0$, the solutions $\hat{c}(x)$ to \eqref{value} and $\hat{Y}(y)$ to \eqref{dual_value} exist and are unique and, if $y=u'(x)$, we have the dual relations
\[
\hat{Y}_t(y)(\omega) = U'\bigl(t,\omega,\hat{c}_t(x)(\omega)\bigr),
\qquad \ud\kappa\times\PP\text{-a.e.},
\]
and
\[
\mathbb{E}\left[\int_0^{\infty}\hat{c}_t(x)\hat{Y}_t(y)\,\ud\kappa_t\right] = xy.
\]
Finally, the dual value function $v$ can be represented as
\be	\label{eq:v_defl}
v(y) = \inf_{Z\in\mathcal Z}\expec\left[\int_0^{\infty}V(t,\omega,yZ_t)\,\ud\kappa_t\right],\quad y>0.
\ee
\end{thm}

\begin{rem}\label{bipolar}
For $\kappa$ corresponding to  maximization of utility from terminal wealth, it can be checked that the sets $\mathcal A$ and $\mathcal Y$ satisfy the assumptions of \cite[Proposition 3.1]{KS99}. This implies that for a deterministic utility $U$ satisfying the Inada conditions and such that  $AE(U) < 1$ (in the terminology of \cite{KS99}), under the additional assumption of finiteness of $u(x)$ for some $x>0$, the assertions of \cite[Theorem 2.2]{KS99} hold under (\ref{NUPBR}) (and possibly without NFLVR). This is a consequence of ``abstract'' Theorems 3.1 and 3.2 in~\cite{KS99} that also apply under \eqref{NUPBR}.
Note also that the condition $u(x)>-\infty$ for all $x>0$ trivially holds if $U$ is a deterministic real-valued utility function. In particular, this is the case in the setting of \cite{KS2003}, where it is shown that the finiteness of the dual function $v$ acts as a necessary and sufficient condition for the validity of the key assertions of the theory.
\end{rem}

\section{Proofs}		\label{sec:proofs}
\begin{proof}[Proof of Proposition \ref{prop:NUPBR}]
Suppose that \eqref{NUPBR} holds. Then, for every $n\in\N$, the set $\X_n$ is bounded in $\LL^0$ and, by \cite[Theorem 2.6]{TS14}, there exists a strictly positive c\`adl\`ag local martingale $Z^n$ such that $Z^n_0=1$ (since $\F_0$ is   trivial) and the $\R^d$-valued process $Z^nS$ is a sigma-martingale on $[0,n]$. As a consequence of \cite[Corollary 3.5]{AnselStricker94} (see also \cite[Remark 2.4]{CDM15}), it holds that $Z^nX$ is a local martingale on $[0,n]$, for every $X\in\X$ and $n\in\N$. For all $t\geq0$, let then $n(t):=\min\{n\in\N : n>t\}$ and define the c\`adl\`ag process $Z=(Z_t)_{t\geq0}$ via
\[
Z_t := \prod_{k=1}^{n(t)}\frac{Z^k_{k\wedge t}}{Z^k_{(k-1)\wedge t}},
\qquad t\geq0.
\]
We now claim that $Z\in\Z$. Since $X\equiv 1\in\X$ and in view of \cite[Lemma I.1.35]{JacodShiryaev03}, it suffices to show that, for every $X\in\X$, the process $ZX$ is a local martingale on $[0,m]$, for each $m\in\N$. 
Fix $m\in\N$. Consider an arbitrary $X\in\X$ and let $\{\tau^n_k\}_{k\in\N}$ be a localizing sequence for the local martingale $Z^nX$ on $[0,n]$, for each $n\in\{1,\ldots,m\}$. Let $\tau^j_{k\;\{\tau^j_k<j\}}:=\tau^j_k\ind_{\{\tau^j_k<j\}}+\infty\ind_{\{\tau^j_k\geq j\}}$, 
for $j=1,\ldots,m$ and $k\in\N$, and define the stopping times 
\[
T^m_k := \min\left\{\tau^1_{k\;\{\tau^1_k<1\}},\ldots,\tau^m_{k\;\{\tau^m_k<m\}},m\right\},
\qquad k\in\N.
\]
Similarly as in \cite[proof of Theorem 4.10]{FontanaOksendalSulem15}, it can be readily verified that the stopped process $(ZX)^{T^m_k}$ is a martingale on $[0,m]$, for all $k\in\N$. Since $\lim_{k\rightarrow+\infty}\PP(T^m_k=m)=1$, this shows that $ZX$ is a local martingale on $[0,m]$. By the arbitrariness of $m$, this proves the claim.

To prove the converse implication, note that, for any $X\in\mathcal{X}$ and $Z\in\mathcal{Z}$, the process $ZX$ is a supermartingale and, hence, for every $T\in \mathbb R_{+}$, it holds that $\mathbb{E}[Z_TX_T]\leq 1$. This shows that the set $Z_T\mathcal{X}_T$ is bounded in $\LL^1$ and, hence, the set $\mathcal{X}_T$ is bounded in $\LL^0$.
\end{proof}


Let us now turn to the proof of Theorem \ref{thm:main}. Together with the abstract results established in \cite[Section 3]{mostovyi2015}, the key step is represented by Lemma \ref{lem:cons} below, which generalizes \cite[Lemma 4.2]{mostovyi2015} by relaxing the no-arbitrage-type requirement into condition \eqref{NUPBR}. 

\begin{lem}	\label{lem:cons}
Let $c$ be a nonnegative optional process and $\kappa$ a stochastic clock. Under assumptions \eqref{clock} and \eqref{NUPBR}, the following conditions are equivalent:
\begin{enumerate}
\item[(i)] $c\in\A$;
\item[(ii)] $\sup_{Z\in\Z}\expec[\int_0^{\infty}c_tZ_t\,\ud\kappa_t]\leq 1$.
\end{enumerate}
\end{lem}
\begin{proof}
If $c \in \A$, there exists an $\R^d$-valued predictable $S$-integrable process $H$ such that 
\[
1 +  \int_0^t {H_u\,\ud S_u}  \ge \int_0^t {c_u\,\ud\kappa_u} \ge 0, 
\qquad t \ge 0.
\]
We define $C_t := \int_0^t {c_u\,\ud\kappa_u}$, $t \ge 0$, and observe that $C$ is an increasing process. For an arbitrary $Z\in\Z$, the process $(\int_0^t {C_{u-}\,\ud Z_u})_{t\geq0}$ is a local martingale and we let $\{\tau_n\}_{n\in\N}$ be a localizing sequence such that $(\int {C_{-}\,\ud Z})^{\tau_n}$ is a uniformly integrable martingale, for every $n\in\N$. 
Using the supermartingale property of $Z(1 +  \int{H\,\ud S})$, we obtain for every $n \in\N$
\[
1 \ge \expec\left[ Z_{\tau_n} \left( 1 + \int_0^{\tau_n} {H_u\,\ud S_u} \right)  \right] 
\ge \expec \left[ Z_{\tau_n} C_{\tau_n} \right]
= \expec\left[\int_0^{\tau_n} {Z_u\,\ud C_u} + \int_0^{\tau_n} {C_{u-}\,\ud Z_u}\right],
\]
where the last equality follows from the integration by parts formula.
Since $\{\tau_n\}_{n\in\N}$ is a localizing sequence for $\int {C_{-}\,\ud Z}$, it holds that $\expec[ \int_0^{\tau_n} {C_{u-}\,\ud Z_u}] =0$, for every $n\in\N$. 
Hence: 
\[
1 \ge \expec\left[\int_0^{\tau_n} {Z_u\,\ud C_u}\right],
\qquad\text{ for every }n\in\N.
\]
By the monotone convergence theorem, we get that 
\[
1 \ge \lim_{n\to \infty} \expec\left[\int_0^{\tau_n} {Z_u\,\ud C_u}\right] 
= \expec\left[\int_0^{\infty} {Z_u\,\ud C_u}\right].
\]
Since $Z\in\Z$ is arbitrary, this proves the implication (i)$\Rightarrow$(ii).

Suppose now that $\sup_{Z\in\Z} \expec[ \int_0^{\infty} {c_tZ_t\,\ud\kappa_t} ] \le 1$. 
Take an arbitrary $Z\in\Z$ and let $\{\varrho_n\}_{n\in\N}$ be a sequence of  bounded  stopping times increasing to infinity $\PP$-a.s., such that $Z^{\varrho_n}$ is a uniformly integrable martingale, for each $n\in\N$. 
Denoting
$
\M_{\sigma}(S) 
:= \bigl\{\QQ\sim\PP\text{ : $S$ is a $\QQ$-sigma-martingale}\bigr\}
$,
one can show that 
$\M_{\sigma}(S^{\varrho_n})\neq\emptyset$, for every $n\in\N$.
Let $\QQ\in\M_{\sigma}(S^{\varrho_n})$ and denote by $M=(M_t)_{t\geq0}$ its c\`adl\`ag density process (i.e., $M_t= \ud\QQ|_{\F_t}/\ud\PP|_{\F_t}$, $t\geq0$). 
Letting $Z':=M^{\varrho_n}Z(Z^{\varrho_n})^{-1}$, 
{\cite[Lemma 2.3]{stricker1998some} implies that $Z'\in\Z$.}
Therefore, for any stopping time $\tau$,
\[
\expec^{\QQ}[C_{\tau\wedge\varrho_n}]
= \expec[M_{\tau\wedge\varrho_n}C_{\tau\wedge\varrho_n}]
= \expec[Z'_{\tau\wedge\varrho_n}C_{\tau\wedge\varrho_n}] \leq 1,
\]
where the last inequality follows from the assumption that $\sup_{Z\in\Z} \expec[ \int_0^{\infty} {c_tZ_t\,\ud\kappa_t} ] \le 1$ by the same arguments used in the first part of the proof together with an application of Fatou's lemma. 
%
As a consequence, we get
\[
\sup_{\QQ\in\M_{\sigma}(S^{\varrho_n})}\sup_{\tau\in\T} \, \expec^{\QQ}[C_{\tau\wedge\varrho_n}] \leq 1,
\]
where $\T$ is the set of all stopping times.
\cite[Proposition 4.2]{FK97} 
then gives the existence of an adapted c\`adl\`ag process $V^n$ such that $V^n_t\geq C_{t\wedge\varrho_n}$, for every $t\geq0$, and admitting a decomposition of the form
\[
V^n_t = V^n_0 + \int_0^tH^n_u\,\ud S^{\varrho_n}_u - A^n_t,
\qquad t\geq0,
\]
where $H^n$ is an $\R^d$-valued predictable $S^{\varrho_n}$-integrable process, $A^n$ is an adapted increasing process with $A^n_0=0$ and $V^n_0=\sup_{\QQ\in\M_{\sigma}(S^{\varrho_n}),\tau\in\T}\, \expec^{\QQ}[C_{\tau\wedge\varrho_n}]\leq1$.
Therefore, for every $n\in\N$, we obtain
\[
1 +  \int_0^t {H^n_u\,\ud S_u}
\geq V^n_0 +  \int_0^t {H^n_u\,\ud S_u}
= V^n_t + A^n_t
\geq V^n_t
\geq C_t,
\qquad 0\leq t\leq \varrho_n.
\]
Let $\bar{H}^n:=H^n\ind_{\dbracc{0,\varrho_n}}$, for all $n\in\N$. By \cite[Lemma 5.2]{FK97}, we can construct a sequence of processes $\{Y^n\}_{n\in\N}$, with $Y^n\in\conv(1+\int\bar{H}^n\,\ud S,1+\int\bar{H}^{n+1}\,\ud S,\ldots)$, $n\in\N$, and a c\`adl\`ag process $Y$ such that $\{ZY^n\}_{n\in\N}$ is Fatou convergent
to a supermartingale $ZY$, for every strictly positive c\`adl\`ag local martingale $Z$ such that $ZX$ is a supermartingale for every $X\in\X$. Note that $Y_t\geq C_t$, for all $t\geq0$, and $Y_0\leq 1$. 
Similarly as above, applying \cite[Theorem 4.1]{FK97} to the stopped process $Y^{\varrho_n}$, for $n\in\N$, we  obtain the decomposition
\[
Y_t^{\varrho_n} = Y_0 + \int_0^tG^n_u\,\ud S_u^{\varrho_n} - B^n_t,
\qquad t \geq 0,
\]
where $G^n$ is an $\R^d$-valued predictable $S^{\varrho_n}$-integrable process and $B^n$ is an adapted increasing process with $B^n=0$, for  $n\in\N$. 
Letting 
\[
G := G^1+\sum_{n=1}^{\infty}(G^{n+1}-G^n)\ind_{\dbraoo{\varrho_n,+\infty}}
= G^1\ind_{\dbracc{0,\varrho_1}}+\sum_{n=1}^{\infty}G^{n+1}\ind_{\dbraoc{\varrho_n,\varrho_{n+1}}},
\] 
it follows that $1+\int_0^tG_u\,\ud S_u\geq C_t$, for all $t\geq0$, thus establishing the implication  (ii)$\Rightarrow$(i).
\end{proof}
 
We are now in a position to complete the proof of Theorem~\ref{thm:main}, which generalizes the results of  \cite[Theorems 2.3 and 2.4]{mostovyi2015} to the case where only \eqref{NUPBR} is assumed to hold. 

\begin{lem}	\label{lem:closedness}
Under \eqref{NUPBR}, the set $\Z$ is closed under countable convex combinations. If in addition \eqref{clock} holds, then 
for every $c\in\A$, we have
\be	\label{eq:cons}
\sup_{Z\in\Z}\expec\left[\int_0^{\infty}c_tZ_t\,\ud\kappa_t\right]
= \sup_{Y\in\Y}\expec\left[\int_0^{\infty}c_tY_t\,\ud\kappa_t\right] \leq 1.
\ee
\end{lem}
\begin{proof}
Let $\{Z^n\}_{n\in\N}$ be a sequence of processes belonging to $\Z$ and $\{\lambda^n\}_{n\in\N}$ a sequence of positive numbers such that $\sum_{n=1}^{\infty}\lambda^n=1$. Letting $Z:=\sum_{n=1}^{\infty}\lambda^nZ^n$, we need to show that $Z\in\Z$.
For each $N\in\N$, define $\widetilde{Z}^N:=\sum_{n=1}^N\lambda^nZ^n$. For every $X\in\X$, $\{\widetilde{Z}^NX\}_{N\in\N}$ is an increasing sequence of nonnegative local martingales (i.e. $\widetilde Z^{N + 1}_tX_t \geq \widetilde Z^N_tX_t$, for all $N\in\N$ and $t\geq0$), such that $\widetilde{Z}^N_t X_t$ converges a.s. to $Z_t X_t$ as $N\rightarrow+\infty$, for every $t\geq0$, 
and $Z_0X_0 = 1$.
The local martingale property of $ZX$ then follows from \cite[Proposition 5.1]{KLP14} (note that its proof carries over without modifications to the infinite horizon case), whereas  \cite[Theorem VI.18]{DelMey82} implies that $ZX$ is a c\`adl\`ag process. Since $X\in\X$ is arbitrary and $X\equiv 1\in\X$, this proves the claim.
Relation \eqref{eq:cons} follows by the same arguments used in \cite[Lemma 4.3]{mostovyi2015}.
\end{proof}

We denote by $\LL^0(\ud\kappa\times\PP)$ the linear space of equivalence classes of real-valued optional processes on the stochastic basis $(\Omega,\F,(\F_t)_{t\in[0,\infty)},\PP)$, equipped with the topology of convergence in measure $(\ud\kappa\times\PP)$. Let $\LL^0_+(\ud\kappa\times\PP)$ be the positive orthant of $\LL^0(\ud\kappa\times\PP)$.

\begin{proof}[Proof of Theorem~\ref{thm:main}]
The sets $\A$ and $\Y$ are convex solid subsets of $\LL^0_+(\ud\kappa\times\PP)$. By definition, $\Y$ is closed in the topology of convergence in measure $(\ud\kappa\times\PP)$. 
A simple application of Fatou's lemma together with Lemma \ref{lem:cons} allows to show that $\A$ is also closed in the same topology. Moreover, by the same arguments used in \cite[part (ii) of Proposition 4.4]{mostovyi2015}, Lemma \ref{lem:cons} and the bipolar theorem of \cite{BrannathSchachermayer99} imply that $\A$ and $\Y$ satisfy the bipolar relations
\begin{align}
c\in\A &\qquad\Longleftrightarrow\qquad
\expec\left[\int_0^{\infty}c_tY_t\,\ud\kappa_t\right]\leq 1
\quad\text{ for all }Y\in\Y,	\label{eq:bipolar1}\\
Y\in\Y &\qquad\Longleftrightarrow\qquad
\expec\left[\int_0^{\infty}c_tY_t\,\ud\kappa_t\right]\leq 1
\quad\text{ for all }c\in\A.	\label{eq:bipolar2}
\end{align}
Since $X\equiv 1\in\X$ and $\Z\neq\emptyset$, both $\A$ and $\Y$ contain at least one strictly positive element. 
In view of Lemma \ref{lem:closedness}, Theorem \ref{thm:main} then follows directly from \cite[Theorems 3.2 and 3.3]{mostovyi2015}. 
\end{proof}

\begin{rem}
We want to mention that Theorem~\ref{thm:main} can also be proved by means of a change-of-num\'eraire argument. Indeed, one can consider the market where quantities are denominated in units of the {\em num\'eraire portfolio} (whose existence is equivalent to NUPBR, see \cite{KaratzasKardaras07}) and apply \cite[Theorems 2.3 and 2.4]{mostovyi2015} directly in that market, for which the set \eqref{NFLVF} is non-empty. In this regard, see \cite[Section~4.7]{KaratzasKardaras07} and \cite{Kardaras13} in the case of maximization of expected (deterministic) utility from terminal wealth.
\end{rem}

\subsection*{Acknowledgement}{Huy N. Chau's work is supported by Natixis Foundation for Quantitative Research and by the Lend\"ulet grant LP2015-6 of the Hungarian Academy of Sciences. Oleksii Mostovyi's work is supported by the National Science Foundation under award number DMS-1600307.}

\bibliographystyle{alpha}
\bibliography{note_mostovyi_2}

\begin{thebibliography}{KLPO14}

\bibitem[ACDJ14]{ACDJ14}
A.~Aksamit, T.~Choulli, J.~Deng, and M.~Jeanblanc.
\newblock Non-arbitrage up to random horizon for semimartingale models.
\newblock Preprint (available at \texttt{http://arxiv.org/abs/1310.1142}),
  2014.

\bibitem[AS94]{AnselStricker94}
J.P. Ansel and C.~Stricker.
\newblock Couverture des actifs contingents et prix maximum.
\newblock {\em Annales de l'Institut Henri Poincar{\'e}}, 30(2):303--315, 1994.

\bibitem[BS99]{BrannathSchachermayer99}
W.~Brannath and W.~Schachermayer.
\newblock A bipolar theorem for ${L}^0_+({\Omega},\mathcal{{F}},{P})$.
\newblock In J.~Az{\'e}ma, M.~{\'E}mery, M.~Ledoux, and M.~Yor, editors, {\em
  S{\'e}minaire de Probabilit{\'e}s XXXIII}, volume 1709 of {\em Lecture Notes
  in Mathematics}, pages 349--354. Springer, Berlin - Heidelberg, 1999.

\bibitem[CDM15]{CDM15}
T.~Choulli, J.~Deng, and J.~Ma.
\newblock How non-arbitrage, viability and num\'eraire portfolio are related.
\newblock {\em Finance and Stochastics}, 19(4):719--741, 2015.

\bibitem[CL07]{KasperChristensen}
M.~Christensen and K.~Larsen.
\newblock No arbitrage and the growth optimal portfolio.
\newblock {\em Stochastic Analysis and Applications}, 25(1):255--280, 2007.

\bibitem[DM82]{DelMey82}
C.~Dellacherie and P.-A. Meyer.
\newblock {\em Probabilities and Potential B}.
\newblock North-Holland, Amsterdam, 1982.

\bibitem[DS94]{DelbaenSchachermayer94}
F.~Delbaen and W.~Schachermayer.
\newblock A general version of the fundamental theorem of asset pricing.
\newblock {\em Mathematische Annalen}, 300(1):463--520, 1994.

\bibitem[DS95]{DelbaenSchachermayer95a}
F.~Delbaen and W.~Schachermayer.
\newblock Arbitrage possibilities in {B}essel processes and their relations to
  local martingales.
\newblock {\em Probability Theory and Related Fields}, 102(3):357--366, 1995.

\bibitem[DS98]{DelbaenSchachermayer98b}
F.~Delbaen and W.~Schachermayer.
\newblock The fundamental theorem of asset pricing for unbounded stochastic
  processes.
\newblock {\em Mathematische Annalen}, 312(2):215--250, 1998.

\bibitem[FK97]{FK97}
H.~F\"ollmer and D.~Kramkov.
\newblock Optional decomposition under constraints.
\newblock {\em Probability Theory and Related Fields}, 109(1):1--25, 1997.

\bibitem[FK09]{FK09}
R.~Fernholz and I.~Karatzas.
\newblock Stochastic portfolio theory: an overview.
\newblock In A.~Bensoussan and Q.~Zhang, editors, {\em Handbook of Numerical
  Analysis}, volume~15, pages 89--167. North-{H}olland, 2009.

\bibitem[F{\O}S15]{FontanaOksendalSulem15}
C.~Fontana, B.~{\O}ksendal, and A.~Sulem.
\newblock Market viability and martingale measures under partial information.
\newblock {\em Methodology and Computing in Applied Probability}, 17(1):15--39,
  2015.

\bibitem[FR13]{FontanaRunggaldier13}
C.~Fontana and W.J. Runggaldier.
\newblock Diffusion-based models for financial markets without martingale
  measures.
\newblock In F.~Biagini, A.~Richter, and H.~Schlesinger, editors, {\em Risk
  Measures and Attitudes}, EAA Series, pages 45--81. Springer, London, 2013.

\bibitem[HK79]{HarrisonKreps79}
J.M. Harrison and D.M. Kreps.
\newblock Martingales and arbitrage in multiperiod securities markets.
\newblock {\em Journal of Economic Theory}, 20(3):381--408, 1979.

\bibitem[JS03]{JacodShiryaev03}
J.~Jacod and A.N. Shiryaev.
\newblock {\em Limit Theorems for Stochastic Processes}.
\newblock Springer, Berlin - Heidelberg - New York, 2nd edition, 2003.

\bibitem[Kar10]{Kardaras10}
C.~Kardaras.
\newblock Finitely additive probabilities and the fundamental theorem of asset
  pricing.
\newblock In C.~Chiarella and A.~Novikov, editors, {\em Contemporary
  Quantitative Finance: Essays in Honour of {E}ckhard {P}laten}, pages 19--34.
  Springer, Berlin - Heidelberg, 2010.

\bibitem[Kar12]{Kardaras12}
C.~Kardaras.
\newblock Market viability via absence of arbitrage of the first kind.
\newblock {\em Finance and Stochastics}, 16(4):651--667, 2012.

\bibitem[Kar13]{Kardaras13}
C.~Kardaras.
\newblock On the closure in the {E}mery topology of semimartingale
  wealth-process sets.
\newblock {\em Annals of Applied Probability}, 23(4):1355--1376, 2013.

\bibitem[Kar14]{Kar14}
C.~Kardaras.
\newblock A time before which insiders would not undertake risk.
\newblock In Y.~Kabanov, M.~Rutkowski, and T.~Zariphopoulou, editors, {\em
  Inspired by Finance}, pages 349--362. Springer, 2014.

\bibitem[KK07]{KaratzasKardaras07}
I.~Karatzas and K.~Kardaras.
\newblock The numeraire portfolio in semimartingale financial models.
\newblock {\em Finance and Stochastics}, 11(4):447--493, 2007.

\bibitem[KLPO14]{KLP14}
I.~Klein, E.~L\'epinette, and L.~Perez-Ostafe.
\newblock Asymptotic arbitrage with small transaction costs.
\newblock {\em Finance and Stochastics}, 18(4):917--939, 2014.

\bibitem[KLSX91]{KLSX}
I.~Karatzas, J.P. Lehoczky, S.E. Shreve, and G.L. Xu.
\newblock Martingale and duality methods for utility maximization in an
  incomplete market.
\newblock {\em SIAM Journal on Control and Optimization}, 29(3):702--730, 1991.

\bibitem[KS99]{KS99}
D.~Kramkov and W.~Schachermayer.
\newblock The asymptotic elasticity of utility functions and optimal investment
  in incomplete markets.
\newblock {\em Annals of Applied Probability}, 9(3):904--950, 1999.

\bibitem[KS03]{KS2003}
D.~Kramkov and W.~Schachermayer.
\newblock Necessary and sufficient conditions in the problem of optimal
  investment in incomplete markets.
\newblock {\em Annals of Applied Probability}, 13(4):1504--1516, 2003.

\bibitem[K{\v{Z}}03]{KaratzasZitkovic03}
I.~Karatzas and G.~{\v{Z}}itkovi\'{c}.
\newblock Optimal consumption from investment and random endowment in
  incomplete semimartingale markets.
\newblock {\em Annals of Probability}, 31(4):1821--1858, 2003.

\bibitem[Lar09]{Larsen09}
K.~Larsen.
\newblock Continuity of utility-maximization with respect to preferences.
\newblock {\em Mathematical Finance}, 19(2):237--250, 2009.

\bibitem[L{\v{Z}}13]{LarsenZitkovic13}
K.~Larsen and G.~{\v{Z}}itkovi\'{c}.
\newblock On utility maximization under convex portfolio constraints.
\newblock {\em Annals of Applied Probability}, 23(2):665--692, 2013.

\bibitem[Mer69]{Merton69}
R.C. Merton.
\newblock Lifetime portfolio selection under uncertainty: the continuous-time
  case.
\newblock {\em Review of Economics and Statistics}, 51(3):247--257, 1969.

\bibitem[Mos15]{mostovyi2015}
O.~Mostovyi.
\newblock Necessary and sufficient conditions in the problem of optimal
  investment with intermediate consumption.
\newblock {\em Finance and Stochastics}, 19(1):135--159, 2015.

\bibitem[SY98]{stricker1998some}
C.~Stricker and J.-A. Yan.
\newblock Some remarks on the optional decomposition theorem.
\newblock In J.~Az\'ema, M.~Yor, M.~\'Emery, and M.~Ledoux, editors, {\em
  S{\'e}minaire de Probabilit{\'e}s XXXII}, volume 1686 of {\em Lecture Notes
  in Mathematics}, pages 56--66. Springer, Berlin - Heidelberg, 1998.

\bibitem[TS14]{TS14}
K.~Takaoka and M.~Schweizer.
\newblock A note on the condition of no unbounded profit with bounded risk.
\newblock {\em Finance and Stochastics}, 18(2):393--405, 2014.

\bibitem[{\v{Z}}it05]{Zitkovic05}
G.~{\v{Z}}itkovi\'{c}.
\newblock Utility maximization with a stochastic clock and an unbounded random
  endowment.
\newblock {\em Annals of Applied Probability}, 15(1{B}):748--777, 2005.

\bibitem[{\v{Z}}it10]{Zit09}
G.~{\v{Z}}itkovi{\'{c}}.
\newblock Convex-compactness and its applications.
\newblock {\em Mathematics and Financial Economics}, 3(1):1--12, 2010.

\end{thebibliography}

\end{document}